\documentclass[journal]{IEEEtran}
\IEEEoverridecommandlockouts
\usepackage{cite}
\usepackage{amsmath,amssymb,amsfonts}
\usepackage{algorithmic}
\usepackage{graphicx}
\usepackage{textcomp}
\usepackage{xcolor}
\usepackage{amsfonts}
\usepackage{bm}
\usepackage{comment}
\usepackage{booktabs}
\usepackage{algorithm}
\usepackage{multirow,bm,bbm,array,setspace}
\usepackage{textcomp}
\usepackage{psfrag}
\usepackage{pstricks,enumerate}
\usepackage{bbm,subfigure}
\usepackage{dblfloatfix}

\usepackage{array}

\makeatletter
\newcommand{\distas}[1]{\mathbin{\overset{#1}{\kern\z@\sim}}}%
\newsavebox{\mybox}\newsavebox{\mysim}
\newcommand{\distras}[1]{%
  \savebox{\mybox}{\hbox{\kern1pt$\scriptstyle#1$\kern1pt}}%
  \savebox{\mysim}{\hbox{$\sim$}}%
  \mathbin{\overset{#1}{\kern\z@\resizebox{\wd\mybox}{\ht\mysim}{$\sim$}}}%
}
\makeatother
\ifCLASSINFOpdf
\else
\fi

\newcommand{\hh}{\mathrm{H}}
\newcommand{\T}{\mathrm{T}}

\newtheorem{proposition}{Proposition}
\newtheorem{lemma}{Lemma}
\newtheorem{theorem}{Theorem}

\newtheorem{remark}{Remark}

\begin{document}

\title{
Connections Between Quadratic Transform for Fractional Programming and Schur Complement

\author{Kaiming Shen, \IEEEmembership{Senior Member,~IEEE}, Kareem M. Attiah,  \IEEEmembership{Member,~IEEE},\\ Yannan Chen,  \IEEEmembership{Member,~IEEE}, and Wei Yu,  \IEEEmembership{Fellow,~IEEE}
}
\thanks{
Manuscript submitted on \today.
Kaiming Shen is with the School of Science and Engineering, The Chinese University of Hong Kong, Shenzhen, China. 
Yannan Chen is with the College of Electronics and Information Engineering, Shenzhen University, Shenzhen, China. 
Kareem M. Attiah and Wei Yu are with The Edward S. Rogers Sr. Department of Electrical and Computer Engineering, University of Toronto, Toronto, ON M5S3G4, Canada.
(e-mails: shenkaiming@cuhk.edu.cn, kareem.attiah@mail.utoronto.ca, chenyannan@szu.edu.cn, weiyu@ece.utoronto.ca.)

This work has been presented in part at the \emph{IEEE International Symposium on Information Theory (ISIT)}, Guangzhou, China, June 2026 \cite{ShenAttiahChenYu2006}. 
The work of Kaiming Shen and Yannan Chen was supported by the NSFC under Grant 12426306.
The work of Kareem M. Attiah and Wei Yu was supported by an NSERC Discovery Grant RGPIN-2023-04697.
}
}

\maketitle
\begin{abstract}
This paper shows that there are intimate connections between the quadratic
transform technique for solving fractional programming (FP) problems and the
Schur-complement technique in matrix analysis. We demonstrate that the
quadratic transform technique is related to two aspects of the Schur
complement: (i) the linear matrix inequality (LMI) condition for positive
semidefiniteness and (ii) the matrix determinant formula. Specifically, we
establish that the quadratic transform and the Schur-complement LMI condition
imply each other. This connection allows us to provide new interpretations of
the auxiliary variable in the quadratic transform, and it allows us to rederive
the Schur-complement determinant formula. Furthermore, this connection leads
to generalizations of the quadratic transform in FP and the Schur-complement
LMI that can accommodate generalized matrix inverse. 
As an application in information theory, we apply the generalized FP framework 
to the least-favorable-noise minimax formulation of the Gaussian vector broadcast 
channel sum capacity problem. When the least-favorable noise covariance is singular, 
matrix-inverse-based Karush-Kuhn-Tucker (KKT) analysis would require a careful
analysis of the input and output spaces of the channel. 
We show using generalized FP that an auxiliary-variable representation of the singular
matrix fraction directly yields the reciprocal multiple-access channel
and recovers the uplink-downlink duality relation for sum capacity. 

\end{abstract}


\section{Introduction}
\label{sec:intro}

This paper explores a curious connection between two seemingly unrelated topics---the quadratic transform for solving fractional programming (FP) problems \cite{FP_SPM} and the Schur complement in matrix analysis \cite{Fuzhen_book}, both of which have been extensively used in the communication system design. Specifically, we show that the quadratic transform for FP can be obtained from the Schur complement, and vice versa. Furthermore, this connection leads to new interpretations and generalizations of results in FP and Schur complement.

FP is a class of mathematical optimizations that involve ratio terms. It plays a key role in communications and signal processing, because many key performance metrics are fractionally structured, e.g., the signal-to-interference-plus-noise ratio (SINR), Cram\'{e}r-Rao bound, and energy efficiency. While the classical FP methods (including Dinkelbach's algorithm \cite{dinkelbach1967nonlinear}) are typically limited to the single-scalar-ratio case, the more recent development of the quadratic transform for FP \cite{shen2018fractional1} allows us to handle a broader range of problems with multiple ratios and with matrix variables, as motivated by optimization involving multiple-input multiple-output (MIMO) networks. 

The following theorem states a typical result on the quadratic transform for FP.
This result is useful for developing algorithms for optimizing ratios, because 
it decouples the numerator and denominator of each fractional term. 

\begin{theorem}[Quadratic Transform for FP \cite{FP_SPM}] 
\label{thm:FP}
Consider $k$ pairs of matrix functions $A_i:\mathcal X\rightarrow \mathbb C^{n\times m}$ and $B_i:\mathcal X\rightarrow \mathbb H^{n}_{++}$, for $i=1,\ldots,k$, where $\mathcal X$ is a nonempty constraint set, and $\mathbb H^{n}_{++}$ denotes the set of $n\times n$ positive definite matrices.  The sum-of-traces-of-matrix-ratio FP problem
\begin{align}
\label{FP:prob}
\underset{x \in \mathcal X}{\text{maximize}} &\quad \sum^k_{i=1}\operatorname{Tr} \Bigl(A^{\hh}_i(x) B^{-1}_i(x) A_i(x)\Bigr)
\end{align}
is equivalent to
\begin{equation}
\label{QT prob}
\begin{aligned}
\underset{x \in \mathcal X, Y_1,\ldots,Y_k}{\text{\ \ maximize}} &\quad \sum^k_{i=1}\operatorname{Tr}\Bigl(2\Re\bigl\{Y_i^{\hh}A_i(x)\bigr\}-Y_i^{\hh}B_i(x) Y_i\Bigr) \\
\text{subject to} 
&\quad\, Y_i\in\mathbb C^{n\times m},\quad \forall i=1,\ldots,k,
\end{aligned}
\end{equation}
where $\Re\{\cdot\}$ denotes the real part, in the sense that the two problems have the same solution for $x$ and their optimal objective values are equal.
\end{theorem}

The proof of the above theorem is based on explicitly finding the optimal $Y_i$ by completing-the-square, then substituting the optimal $Y_i$ back into the objective function \cite{shen2018fractional1}. This paper shows that there is an alternative way of deriving the same transform based on the Schur complement in matrix analysis. In the rest of the paper, we drop the arguments in the matrix functions, e.g.,
$A_i(x)$ is written as $A_i$, whenever doing so causes no confusion. 

The notion of Schur complement in linear algebra has a long history \cite{Schur}. Consider a square matrix
\begin{equation} \label{eq:M_schur}
M=\begin{bmatrix}
    C & A^{\hh}\\
    A & B
\end{bmatrix}
\end{equation}
with the blocks $A\in\mathbb C^{n\times m}$, $B\in\mathbb H^{n}_{++}$ (with slight abuse of notation), and $C\in\mathbb H^{m}_+$, where $\mathbb H^m_+$ denotes the set of $m\times m$ positive semidefinite matrices. The square matrix $C-A^{\hh}B^{-1}A$ is referred to as the \emph{Schur complement} of the block $B$ of the matrix $M$. 

A classic result concerning the Schur complement is a linear matrix inequality (LMI) condition for the positive semidefiniteness of $C-A^{\hh}B^{-1}A$. 
A proof can be found in \cite{boyd2004convex}. 
\begin{theorem}[Schur Complement and LMI]
\label{thm:Schur}
When $B\succ0$, the following two conditions are equivalent:
\begin{equation}
\label{Schur complement:equivalence}
 C-A^{\hh}B^{-1}A \succeq 0 \quad\Leftrightarrow \quad M\succeq 0.
\end{equation}
\end{theorem}

Thus, the positive semidefiniteness of the Schur complement
is converted to the LMI $M\succeq0$, which is potentially easier to deal with for constrained optimization.
There is also a second important result on the Schur complement concerning the matrix determinant, which will be discussed in Section \ref{sec:det}. 

The Schur complement has seen fruitful applications in information theory, e.g., in Wyner-Ziv coding \cite{WZ_TIT}, 
the relation between the mutual information and the minimum mean-squared error (MMSE) \cite{Guo_TIT}, the dirty-paper coding \cite{WY_ISIT}, and the uplink-downlink duality \cite{Yu2006}. Furthermore, quantum information theory relies heavily on linear algebra; the Schur complement is a useful mathematical tool in this domain \cite{Wilde_book}. 

Recently, both FP and Schur complement have been used in the algorithm design for joint sensing and communications \cite{Ng_TWC, Wang_TWC, Yuan_IOTJ, weizhang_TWC, max_min_FP, attiah2026, 10097000}, indicating that there may be an intimate connection between the two, but such a connection 
has not been properly explored. The present paper aims to develop a unified theory between the two. The main results of the paper are:
\begin{itemize}
    \item We show that Theorem \ref{thm:FP} implies Theorem \ref{thm:Schur}. Further, by extending this connection to the multiple-ratio case, we obtain a generalization of the LMI condition in \eqref{Schur complement:equivalence}.
    \item We show that Theorem \ref{thm:Schur} implies Theorem \ref{thm:FP}. This allows us to devise a generalized quadratic transform that accounts for the generalized inverse.
    \item We show that the auxiliary variable in quadratic transform can be reinterpreted in terms of an MMSE estimation problem and use this connection to rederive the Schur-complement determinant formula.
    \item We apply the generalized FP framework to the least-favorable-noise minimax formulation of the Gaussian vector broadcast-channel (BC) sum capacity problem. When the least-favorable noise covariance is singular, the generalized quadratic transform avoids explicit matrix inverse and directly yields the reciprocal multiple-access-channel (MAC) capacity expression, thereby recovering the uplink-downlink duality for sum capacity.
\end{itemize}

\emph{Notations:}
In this paper, $\mathbb{C}^{n \times m}$ denotes the set of $n \times m$ complex matrices, while $\mathbb H^n$ denotes the set of $n\times n$ Hermitian matrices. Moreover,
$\mathbb{H}^n_{+}$ and $\mathbb{H}^n_{++}$ denote the sets of $n \times n$ Hermitian positive semidefinite and positive definite matrices. We use $(\cdot)^{\T}$ and $(\cdot)^{\hh}$ to represent the transpose and Hermitian transpose of a matrix, and write $M \succeq 0$ (or $M \succ 0$) to indicate that $M$ is positive semidefinite (or positive definite). Furthermore, we denote by $\operatorname{Tr}(\cdot)$ the matrix trace, $|\cdot|$ the determinant, $(\cdot)^+$ the generalized inverse, $I$ the identity matrix (with its dimension indicated by a subscript when necessary), $0$ the zero matrix, $\mathcal{R}(\cdot)$ the range space of a matrix, $\mathcal{N}(\cdot)$ the null space of a matrix, $\Re\{\cdot\}$ the real part of a complex number or complex matrix, $\mathbb{E}[\cdot]$ the expectation operator, $\|\cdot\|_2$ the Euclidean norm, $\mathrm{blkdiag}(\cdot)$ the block diagonal matrix, and $h(\cdot)$ the differential entropy. For a block matrix $X$, $\mathcal P_{\mathrm{blk}}(X)$ denotes the projection of $X$ onto its prescribed block-diagonal blocks. 
Moreover, we use $\mathcal{CN}(m,\Sigma)$ to denote a circularly symmetric complex Gaussian random vector with mean $m$ and covariance $\Sigma$.
All logarithms in this paper are base $e$.

\emph{Organization of the Paper:} In Sections \ref{sec:QT_to_Schur} and \ref{sec:Schur_to_QT},
we show how the quadratic transform for FP can be derived from Schur-complement LMI
relation and vice versa. This gives an interpretation of the auxiliary variable in FP
as a Lagrangian dual variable in the Schur complement reformulation of the FP problem. 
In Section \ref{sec:det}, we give a second interpretation of the auxiliary variable, which gives
rise to a connection between FP and the Schur determinant formula. 
The connection between FP and Schur complement allows a generalization of FP to the case
where the denominator matrix is singular. The generalized FP is formally developed in
Section \ref{sec:generalized FP}.  In Section \ref{sec:duality}, we utilize this 
generalized FP to simplify the derivation of the minimax formulation of the sum
capacity of the Gaussian vector broadcast channel. 
Finally, conclusions are drawn in Section \ref{sec:conclude}.

\section{From Quadratic Transform to\\ Schur Complement}
\label{sec:QT_to_Schur}

We start by rederiving the equivalence relation in \eqref{Schur complement:equivalence} from Theorem \ref{thm:FP}. 
Let $m=1$ in Theorem \ref{thm:FP}, so each $A_i:\mathcal X\rightarrow\mathbb C^{n\times m}$ is just a vector function $a_i:\mathcal X\rightarrow\mathbb C^n$. Theorem \ref{thm:FP} states that the sum-of-ratios FP problem
\begin{align}
\underset{x \in \mathcal X}{\text{maximize}} &\quad \sum^k_{i=1}a^{\hh}_iB^{-1}_ia_i
\end{align}
is equivalent to
\begin{align}
\label{vector:FP}
\underset{x \in \mathcal X,\ y_1,\ldots,y_k \in \mathbb C^n}{\text{maximize}} &\quad \sum^k_{i=1}\Bigl(2\Re\{y_i^{\hh}a_i\}-y_i^{\hh}B_i y_i\Bigr).
\end{align}
We can now show the following chain of equivalences: 
\begin{align}
&C-A^{\hh}B^{-1}A \succeq 0\notag\\
&\ \overset{(a)}{\Leftrightarrow} q^{\hh}(C-A^{\hh}B^{-1}A)q\ge0,\; \forall q\in\mathbb C^m\notag\\
&\ \overset{(b)}{\Leftrightarrow} q^{\hh}Cq - \sup_{y\in\mathbb C^n}\{2\Re\{y^{\hh}Aq\}-y^{\hh}By\}\ge0,\;\forall q\notag\\
&\ \Leftrightarrow q^{\hh}Cq - y^{\hh}Aq-(Aq)^{\hh}y+y^{\hh}By\ge0,\;\forall q, y\notag\\
&\ \Leftrightarrow 
\begin{bmatrix}
    q^{\hh} & -y^{\hh}
\end{bmatrix}
\begin{bmatrix}
    C & A^{\hh}\\
    A & B
\end{bmatrix}
\begin{bmatrix}
    q\\
    -y
\end{bmatrix}\ge0,\;\forall q,y\notag\\
&\ \Leftrightarrow 
u^{\hh}
\begin{bmatrix}
    C & A^{\hh}\\
    A & B
\end{bmatrix}
u\ge0,\;\forall u=\begin{bmatrix}
    q\\ -y
\end{bmatrix}\in\mathbb C^{m+n} \notag \\ 
&\ \overset{(c)}{\Leftrightarrow} \ M\succeq0,
\label{chain:FP to Schur}
\end{align}
where $(a)$ follows from the definition of positive semidefiniteness, $(b)$ follows from Theorem \ref{thm:FP} with $m=1$ as mentioned earlier and $k=1$, and $(c)$ follows from the definition of positive semidefiniteness.
Thus, Theorem \ref{thm:Schur} follows from Theorem \ref{thm:FP}. 

It is interesting to note that only the single-ratio FP (with $k=1$) is needed to prove \eqref{Schur complement:equivalence}. It turns out that if we utilize the multiple-ratio result of Theorem \ref{thm:FP}, the above connection between FP and the Schur complement would imply the following more general chain of equivalences: 
\begin{align*}
&C-\sum^k_{i=1}A_i^{\hh}B_i^{-1}A_i \succeq 0\\
&\Leftrightarrow q^{\hh}\Bigg(C-\sum^k_{i=1}A_i^{\hh}B_i^{-1}A_i\Bigg)q\ge0,\;\forall q\in\mathbb C^m\\
&\Leftrightarrow q^{\hh}Cq - \sum^k_{i=1}\sup_{y_i}\{y_i^{\hh}A_iq+(A_iq)^{\hh}y_i-y_i^{\hh}B_iy_i\}
\ge0,\;\forall q\\
&\Leftrightarrow q^{\hh}Cq - \sum^k_{i=1}(y_i^{\hh}A_iq+(A_iq)^{\hh}y_i-y_i^{\hh}B_iy_i)\ge0,\;\forall q, y_i\\
&\Leftrightarrow 
\begin{bmatrix}
    q\\
    -y_1\\
    -y_2\\
    \vdots\\
    -y_k
\end{bmatrix}^{\hh}
\begin{bmatrix}
    C & A_1^{\hh} & A_2^{\hh} & \cdots & A_k^{\hh}\\
    A_1 & B_1 & 0 & \cdots & 0 \\
    A_2 & 0 & B_2 &   & \\
    \vdots & \vdots & &\ddots\\
    A_k & 0 & & & B_k
\end{bmatrix}
\begin{bmatrix}
    q\\
    -y_1\\
    -y_2\\
    \vdots\\
    -y_k
\end{bmatrix}\ge0, \forall q, y_i
\end{align*}
where all the blank blocks are filled with zero matrices. 

The resulting generalization of the Schur-complement LMI is stated in the following theorem.


\begin{theorem}[Generalized Schur Complement LMI]
    \label{thm:generalized LMI}
Consider the matrix
\begin{equation}
    M^{(k)} = 
\begin{bmatrix}
    C & A_1^{\hh} & A_2^{\hh} & \cdots & A_k^{\hh}\\
    A_1 & B_1 & 0 & \cdots & 0 \\
    A_2 & 0  & B_2 &   &  \\
    \vdots & \vdots & &\ddots & \\
    A_k & 0 &  & & B_k
\end{bmatrix}
\end{equation}
with the blocks $A_i\in\mathbb C^{n\times m}$, $B_i\in\mathbb H^{n}_{++}$, and $C\in\mathbb H^{m}$. The following equivalence relation holds:
\begin{equation}
C-\sum^k_{i=1}A_i^{\hh}B_i^{-1}A_i \succeq 0  \quad\Leftrightarrow \quad M^{(k)}\succeq 0.
\end{equation}
Thus, by utilizing the multiple-ratio quadratic transform result in FP, a more general version of the Schur-complement LMI relation can be established. 
\end{theorem}

\begin{remark}
Theorem \ref{thm:generalized LMI} can also be viewed as a direct application of the
standard Schur-complement LMI to a block-diagonal matrix. To
see this, define
\begin{equation}
\hat B=\mathrm{blkdiag}(B_1,\ldots,B_k)\quad\text{and}\quad
\hat A=
\begin{bmatrix}
A_1\\
\vdots\\
A_k
\end{bmatrix}.
\end{equation}
Then the matrix $M^{(k)}$ in $(8)$ can be written compactly as
\begin{equation}
M^{(k)}
=
\begin{bmatrix}
C & \hat A^\hh\\
\hat A & \hat B
\end{bmatrix}.
\end{equation}
Since $\hat B\succ 0$, the standard Schur-complement condition
gives
\begin{equation}
M^{(k)}\succeq 0
\quad\Leftrightarrow\quad
C-\hat A^{\hh}\hat B^{-1}\hat A\succeq 0.
\end{equation}
Moreover,
\begin{equation}
\hat A^{\hh}\hat B^{-1}\hat A
=
\sum_{i=1}^k A_i^{\hh} B_i^{-1} A_i.
\end{equation}
Thus, Theorem \ref{thm:generalized LMI} boils down to a block-diagonal form of the classical Schur-complement LMI. Its relevance here is
that this block-diagonal form arises naturally from the multiple-ratio
quadratic transform for FP.
\end{remark}

\section{From Schur Complement to\\ Quadratic Transform}
\label{sec:Schur_to_QT}

We now discuss the reverse direction and show that Theorem \ref{thm:FP} follows from Theorem \ref{thm:Schur}. In the process, we provide an interpretation of the auxiliary variable $Y_i$ in the quadratic transform for FP.

To ease notation, we start with the single-ratio case of \eqref{FP:prob}:
\begin{align}
\label{single FP:prob}
\underset{x \in \mathcal X}{\text{maximize}} &\quad \operatorname{Tr} (A^{\hh}B^{-1}A). 
\end{align}
The above problem can be readily rewritten as
\begin{subequations}
\label{single FP:prob new}
\begin{align}
\max_{x \in \mathcal X} \min_{C\in\mathbb H^m} &\quad \operatorname{Tr} (C)\\
  \text{subject to}\,\,\, 
  &\quad\;\,C-A^{\hh}B^{-1}A\succeq 0.
\label{single FP:prob new:constraint c}
\end{align}
\end{subequations}
The above problem reformulation rewrites the objective in terms of a Schur complement. By Theorem \ref{thm:Schur}, we convert the positive semidefiniteness constraint \eqref{single FP:prob new:constraint c} to LMI as
\begin{subequations}
\label{single FP:prob LMI}
\begin{align}
\max_{x \in \mathcal X} \min_{C\in\mathbb H^m} &\quad \operatorname{Tr} (C)\\
  \text{subject to}\,\,\, 
  &\quad\; 
  M\succeq 0,
\end{align}
\end{subequations}
where $M$ is defined as in \eqref{eq:M_schur}.
Note that problem \eqref{single FP:prob LMI} consists of two parts: the inner minimization over $C$ for fixed $x$ and the outer maximization over $x$. We proceed to find the Lagrangian dual of the inner problem.

The Lagrangian function of the inner minimization is
\begin{align}
    L(C,\Lambda) = \operatorname{Tr}(C) - \operatorname{Tr}(\Lambda M),
\end{align}
where $\Lambda\in\mathbb H^{m+n}_+$ is the Lagrange multiplier. This $\Lambda$ can be partitioned into four blocks:
\begin{equation}
\label{Lambda}
\Lambda  
=
\begin{bmatrix}
    W & V^{\hh}\\
    V & Z
\end{bmatrix}\succeq 0,
\end{equation}
where $W\in\mathbb H^m_+$, $Z\in\mathbb H^n_+$, and $V\in\mathbb C^{n\times m}$. Substituting \eqref{Lambda} into $L(C,\Lambda)$, we have
\begin{align}
&L(C,\Lambda) = \operatorname{Tr}(C) - \operatorname{Tr}(WC+V^{\hh}A+A^{\hh}V)-\operatorname{Tr}(ZB)\notag\\
    &= \operatorname{Tr}\bigl((I_m-W)C-V^{\hh}A-A^{\hh}V\bigr)-\operatorname{Tr}(ZB).
\end{align}
The dual function is now given by
\begin{align}
    g(\Lambda) &= \inf_{C\in\mathbb H^m} L(C,\Lambda)\\
    &=
    \begin{cases}
      -\operatorname{Tr}\bigl(2\Re\{V^{\hh}A\}\big)-\operatorname{Tr}(ZB) & \text{if}\; W=I_m\vspace{0.3em}\\
      -\infty & \text{otherwise}.
    \end{cases} 
\end{align}
To ensure that the value of $g(\Lambda)$ is finite, we must have 
\begin{equation}
    W = I_m.
\end{equation}
After plugging $W=I_m$ into $\Lambda$ in \eqref{Lambda} and by using the following chain of equivalence relations:
\begin{align}
&\begin{bmatrix}
    I_m & V^{\hh}\\
    V & Z
\end{bmatrix}\succeq 0 \notag\\
&\overset{(a)}{\Leftrightarrow}
\begin{bmatrix}
 0   & I_n\\
 I_m & 0
\end{bmatrix}
\begin{bmatrix}
    I_m & V^{\hh}\\
    V & Z
\end{bmatrix}
\begin{bmatrix}
 0   & I_m\\
 I_n & 0
\end{bmatrix}
\succeq 0\notag\\
&\Leftrightarrow
\begin{bmatrix}
    Z & V\\
    V^{\hh} & I_m
\end{bmatrix}\succeq 0\notag\\
&\overset{(b)}{\Leftrightarrow}
Z-VV^{\hh}\succeq0,
\end{align}
where $(a)$ is due to a unitary congruence relation using a permutation matrix and $(b)$ follows from Theorem \ref{thm:Schur},  we arrive at the dual problem of the inner minimization in \eqref{single FP:prob LMI}: 
\begin{subequations}
\begin{align}
\underset{V,\, Z}{\text{maximize}} &\quad -\operatorname{Tr}\bigl(V^{\hh}A+A^{\hh}V\bigr)-\operatorname{Tr}(ZB)\\
  \text{subject to} &\quad\;\,Z-VV^{\hh}\succeq0.
\end{align}
\end{subequations}
Because $B\succ0$, the optimal solution of $Z$ is
\begin{equation}
    Z^\star = VV^{\hh}.
\end{equation}
As a consequence, the dual problem reduces to
\begin{align}
\underset{V\in\mathbb C^{n\times m}}{\text{maximize}} &\quad -\operatorname{Tr}\bigl(V^{\hh}A+A^{\hh}V+V^{\hh}BV\bigr).
\end{align}
This dual problem is formulated with respect to the inner minimization problem in \eqref{single FP:prob LMI}. Because the inner minimization problem is convex with strong duality, it is equivalent to its dual. Replacing the inner minimization by the dual problem, we convert \eqref{single FP:prob LMI} to
\begin{align}
\underset{x \in \mathcal X,\, V\in\mathbb C^{n\times m}}{\text{maximize}} &\quad -\operatorname{Tr}\bigl(V^{\hh}A+A^{\hh}V+V^{\hh}BV\bigr).
\end{align}
Further, substituting $V=-Y$ into the above problem yields
\begin{align}
\label{singl ratio:QT prob}
\underset{x \in \mathcal X,\, Y\in\mathbb C^{n\times m}}{\text{maximize}} &\quad \operatorname{Tr}\bigl(Y^{\hh}A+A^{\hh}Y-Y^{\hh}BY\bigr),
\end{align}
which is exactly the single-ratio case of the quadratic transform in Theorem \ref{thm:FP}. 

The above result can be immediately extended to multiple ratios. The main idea is to rewrite \eqref{FP:prob} as
\begin{subequations}
\label{}
\begin{align}
\max_{x \in \mathcal X} \min_{C_i\in\mathbb H^m} &\quad \sum^k_{i=1}\operatorname{Tr} (C_i)\\
  \text{subject to}\,\,\, 
  &\quad\;C_i-A^{\hh}_iB^{-1}_iA_i\succeq 0,\quad\forall i=1,\ldots,k.
\label{}
\end{align}
\end{subequations}
After each constraint $C_i-A^{\hh}_iB^{-1}_iA_i\succeq 0$ is converted to an LMI, a Lagrange multiplier $\Lambda_i$ is introduced for each LMI, and can be further determined as
\begin{equation}
    \Lambda_i=
    \begin{bmatrix}
    I & V^{\hh}_i\\
    V_i & V_iV^{\hh}_i
\end{bmatrix}.
\end{equation}
By substituting the resulting dual problem into the inner minimization, we arrive exactly at the multiple-ratio quadratic transform in \eqref{QT prob}. 
So Theorem \ref{thm:Schur} implies Theorem \ref{thm:FP}.

\begin{remark}
It can be observed that the auxiliary variable $Y$ of the quadratic transform in Theorem \ref{thm:FP} is precisely the negative of the off-diagonal block $V$ of the Lagrange multiplier $\Lambda$. 
\end{remark}

\begin{remark}
The result that Theorem \ref{thm:FP} implies Theorem \ref{thm:Schur} and vice versa pertains to \emph{maximization FP}. A slightly weaker connection can also be established between 
Schur complement and \emph{minimization FP}. 
Consider the following FP problem of minimizing the trace of a matrix ratio, as opposed to the earlier maximization problem \eqref{single FP:prob}:
\begin{align}
\label{single FP:prob min}
\underset{x \in \mathcal X}{\text{minimize}} &\quad \operatorname{Tr} (A^{\hh}B^{-1}A). 
\end{align}
Again, by introducing an auxiliary variable $C\in\mathbb H^m_+$, we formulate an optimization problem with a constraint in the form of a Schur complement:
\begin{subequations}
\label{single FP:prob min new}
\begin{align}
\underset{x\in \mathcal X,\, C\in\mathbb H^{m}_{+}}{\text{minimize}} &\quad \operatorname{Tr} (C)\\
  \text{subject to}\,\, 
  &\quad\;C-A^{\hh}B^{-1}A\succeq 0.
\label{}
\end{align}
\end{subequations}
However, there is a crucial difference between the minimization FP and the maximization FP. 
For the minimization FP, after taking the Lagrangian dual of \eqref{single FP:prob min new} and repeating the same procedure as before, we obtain 

\begin{align}
\label{min FP:max-min}
\max_{Y\in\mathbb C^{n\times m}}\;\min_{x \in \mathcal X}&\quad \operatorname{Tr}\bigl(2\Re\{Y^{\hh}A\}-Y^{\hh}BY\bigr).
\end{align}
In contrast, if we directly decouple the matrix ratio $A^{\hh}B^{-1}A$ using the quadratic transform of Theorem \ref{thm:FP}, the minimization FP problem \eqref{single FP:prob min} is recast into
\begin{align}
\min_{x \in \mathcal X}\;\max_{Y\in\mathbb C^{n\times m}}\;\operatorname{Tr}\bigl(2\Re\{Y^{\hh}A\}-Y^{\hh}BY\bigr).
\label{min FP:min-max}
\end{align}
The two can be the same, but only if the minimization and the maximization can be interchanged. 
Thus, there is also a connection between the minimization FP and Schur complement just as in the maximization FP case, but it requires additional conditions that ensure min-max is equal to max-min, e.g., Sion's minimax condition \cite{Sion_minimax}.
\end{remark}

The above weaker form of connection between the Lagrangian dual of the Schur complement and FP is foreshadowed in \cite{Attiah_ISIT24} for the scalar-ratio FP case of \eqref{single FP:prob min} with $m=1$ in the context of an integrated sensing and communication problem. The proof technique of \cite{Attiah_ISIT24} uses the Schur complement, while earlier work \cite{10097000} for a similar problem setting uses the quadratic transform for FP. This indicates a potential connection between the two. 

Finally, we remark that the quadratic transform can be extended to more complicated optimization problems, but 
they do not necessarily have a connection with Schur complement. Examples include the mixed-max-and-min FP and the log-ratio FP as treated in \cite{FP_SPM}.

\section{Connection to Schur Complement Determinant Formula}
\label{sec:det}

This paper aims to provide insight on the interpretations of the auxiliary variable $Y_i$ in the quadratic transform for FP in Theorem \ref{thm:FP}. 
One such interpretation is that, as seen in Section \ref{sec:Schur_to_QT}, the auxiliary variable $Y_i$ corresponds to the negative of the off-diagonal block of the Lagrangian dual variable in a Schur complement reformulation of the original FP problem. 

There is also an alternative interpretation of the auxiliary variable. Recall that in Section \ref{sec:QT_to_Schur} the role of $Y_i$ can be seen in the chain of equivalences \eqref{chain:FP to Schur}, where for the single-ratio case, $y$ is the vector that optimizes a quadratic form as shown in line (b) of \eqref{chain:FP to Schur}. 
As we show in this section, if this optimization of the quadratic form is interpreted as an MMSE estimation, it can give a way of rederiving the Schur-complement determinant formula.
This gives another connection between FP and Schur complement.


Consider an MMSE estimation problem on a jointly circularly symmetric complex Gaussian random vector, 
denoted using $\mathcal{CN}(\cdot,\cdot)$, called $u$:
\begin{equation}
    u=
    \begin{bmatrix}
        q\\
        -y
    \end{bmatrix}
    \sim\mathcal{CN}(0,J),\;\text{where}\; 
J= 
\begin{bmatrix}
    F & G^{\hh}\\
    G & H
\end{bmatrix}\succ 0
\end{equation}
with $F\in\mathbb H^m_{++}$, $G\in\mathbb C^{n\times m}$, and $H\in\mathbb H^n_{++}$. 
Treat $q$ as the observation. Consider the problem of estimating $y$ based on $q$. 
The MMSE estimator is given by
\begin{align}
\hat y &= \mathbb E[y|q]\notag\\
    &= \Sigma_{yq} \Sigma^{-1}_{qq} q = -\,GF^{-1}q,
\end{align}
where $\Sigma_{yq}$ and $\Sigma_{qq}$ are the respective covariance matrices. 

In the meanwhile, in line (b) of \eqref{chain:FP to Schur}, the optimal $y$ is
\begin{equation}
    y^\star = B^{-1}Aq.
\end{equation}
It turns out that if we identify $J=M^{-1}$, then $y^\star$ and $\hat y$ coincide, i.e., the optimal auxiliary variable $y$ of the quadratic transform amounts to an MMSE estimator. This is because the $(2,1)$-th block of $MJ$ is $AF+BG=0$ when $J=M^{-1}$, thus we have $B^{-1}A=-\,GF^{-1}$. 

The above interpretation of $y$ in terms of an MMSE estimation problem gives a way of rederiving the Schur-complement determinant formula. 
Assuming $J=M^{-1}$, the entropy of $u$ can be computed as
\begin{equation}
\label{entroy:1}
    h(u) = \log\Bigl((\pi e)^{m+n}|M^{-1}|\Bigr).
\end{equation}
Next, let us compute $h(u)$ in another way. Denote by $\xi$ the MMSE error in the above estimation problem:
\begin{equation}
    \xi = y - y^\star=y-B^{-1}Aq.
\end{equation}
Because $M=\begin{bmatrix}
    C & A^{\hh}\\
    A & B
\end{bmatrix}$ and $u=\begin{bmatrix}
    q\\ -y
\end{bmatrix}$, it holds that
\begin{align}
u^{\hh}Mu
&= q^{\hh}Cq - y^{\hh}Aq-(Aq)^{\hh}y+y^{\hh}By\notag\\
&=
\begin{bmatrix}
    q^{\hh}  & -\xi^{\hh}
\end{bmatrix}
\begin{bmatrix}
    C-A^{\hh}B^{-1}A & 0 \\
    0 & B
\end{bmatrix}
\begin{bmatrix}
    q  \\ -\xi
\end{bmatrix}.
\end{align}
Since $u\sim\mathcal{CN}(0,M^{-1})$, we have
\begin{equation}
    \begin{bmatrix}
        q\\
        -\xi
    \end{bmatrix}
    \sim\mathcal{CN}\left(0,
\begin{bmatrix}
    C-A^{\hh}B^{-1}A & 0 \\
    0 & B
\end{bmatrix}^{-1}\right),
\end{equation}
from which we see that 
\begin{equation}
\label{distribution: q and tilde y}
    q \sim\mathcal{CN}(0,(C-A^{\hh}B^{-1}A)^{-1}) 
    \;\;\text{and}\;\;
    \xi \sim\mathcal{CN}(0,B^{-1}),
\end{equation}
and that $q$ is independent of $\xi$, which is expected because the MMSE error $\xi$ must be orthogonal to the observation $q$ by the MMSE estimation theory.
Then, 
\begin{align}
\label{entropy:2}
h(u)& \overset{(a)}{=} h(q) + h(\xi)\notag\\
&\overset{(b)}{=} \log\Bigl((\pi e)^m\bigl|(C-A^{\hh}B^{-1}A)^{-1}\bigr|\Bigr)\notag\\
&\qquad+ \log\Bigl((\pi e)^n\bigl|B^{-1}\bigr|\Bigr),
\end{align}
where 
$(a)$ follows since $q$ and $\xi$ are independent, and $(b)$ follows from \eqref{distribution: q and tilde y}. 
Combining \eqref{entroy:1} and \eqref{entropy:2} establishes the Schur-complement determinant formula, assuming $M\succ 0$. 

\begin{theorem}[Schur Complement Determinant Formula]
\label{thm:Schur formula}
When $M\succeq0$ and $B\succ0$, we have
\begin{equation}
\label{thm:Schur:eq}
|M| = |B|\cdot\left|C-A^{\hh}B^{-1}A\right|.
\end{equation}
\end{theorem}

\begin{IEEEproof}
The entropy-based argument leading to the theorem establishes the formula for the case of $M\succ0$. For the general case with $M\succeq0$, we apply the result to
\begin{equation}
M(\tau)=
\begin{bmatrix}
C+\tau I & A^\hh\\
A & B
\end{bmatrix},
\end{equation}
where $\tau>0$ is sufficiently large so that $M(\tau)\succ0$. Then
\begin{equation}
|M(\tau)|
=
|B|\cdot\left|C+\tau I-A^\hh B^{-1}A\right|
\end{equation}
holds for all sufficiently large $\tau$. Since both sides are polynomials in $\tau$ and they agree for $\tau$ larger than a threshold, the identity must hold for all $\tau$, including $\tau=0$.
\end{IEEEproof}


\section{FP with Generalized Inverse}
\label{sec:generalized FP}

This section extends the connection between FP and Schur complement to the case where the denominator matrix is singular. The development proceeds from the generalized Schur complement to a discussion of the range condition and finally to the Loewner monotone generalization.


\subsection{Generalized Schur Complement}
\label{subsec:generalized FP}

The connection between the Schur complement and the quadratic transform allows
advances in one area to be mapped to the other area.  Specifically, the Schur
complement can be extended to the case of the so-called generalized inverse
\cite{Fuzhen_book}. This means that the quadratic transform technique for
solving FP can also be extended for the generalized inverse. 
In this section, we develop such an extension of Theorem \ref{thm:FP}. 

For any matrix $D\in\mathbb C^{m\times n}$, its generalized inverse, denoted by $D^+$, is any $n\times m$ matrix such that 
\begin{equation}
DD^+D=D. 
\end{equation}
The Moore-Penrose inverse is a special case of the generalized inverse. The following is a result on the Schur complement for the generalized inverse \cite{Fuzhen_book}.

\begin{theorem}[Schur Complement Based LMI with Generalized Inverse \cite{Fuzhen_book}]
\label{thm:Schur:generalized inverse}
When $B\succeq0$ and $(I_n-BB^+)A=0$, we have 
\begin{equation}
\label{Schur complement:equivalence:generalized}
C-A^{\hh}B^+A \succeq 0 \quad\Leftrightarrow \quad  M\succeq 0.
\end{equation}
\end{theorem}

This theorem requires a new condition $(I_n-BB^+)A=0$, which turns out to be the condition needed in order to ensure that the Schur complement $C-A^{\hh}B^+A$ is well defined. 

Now, let $A:\mathcal X\rightarrow\mathbb C^{n\times m}$ and $B:\mathcal X\rightarrow\mathbb H^{n}_+$ be a pair of matrix functions, where $B(x)$ is not necessarily invertible. 
Consider an FP involving generalized inverse as formulated below, assuming 
the condition $(I_n-BB^+)A=0$, 
\begin{align}
\label{FP:prob:generalized inverse_1}
\underset{x\in \mathcal X}{\text{maximize}} &\quad \operatorname{Tr} (A^{\hh}B^{+}A). 
\end{align}
The above problem can be rewritten as 
\begin{subequations}
\label{generalized inverse FP:prob new}
\begin{align}
\underset{x \in \mathcal X}{\text{maximize}} &\quad \inf_{C\in\mathbb H^{m}}\operatorname{Tr} (C)\\
  \text{subject to} 
  &\quad\;\,C-A^{\hh}B^{+}A\succeq 0.
\label{}
\end{align}
\end{subequations} By Theorem \ref{thm:Schur:generalized inverse}, the
constraint can be replaced by $M \succeq 0$. In this case, we can then use the same Lagrangian procedure to
find its dual, which would lead to an extension of quadratic transform for generalized inverse. 
Observe that the replacement of $B^{-1}$ by $B^{+}$ does not impact the Lagrangian procedure, 
so we would arrive at the same reformulation as in \eqref{singl ratio:QT prob}. The above result can be readily extended to the multiple-ratio case, as stated in the following theorem.

\begin{theorem}[Quadratic Transform for FP with Generalized Inverse] 
\label{thm:FP:generalized inverse}
Consider $k$ pairs of matrix functions $A_i:\mathcal X\rightarrow \mathbb C^{n\times m}$ and $B_i:\mathcal X\rightarrow \mathbb H^{n}_{+}$, for $i=1,\ldots,k$. In particular, assume that $(I_n-B_iB^+_i)A_i=0$ for each $i$. The generalized sum-of-traces-of-matrix-ratio FP problem
\begin{align}
\label{FP:prob:generalized inverse_2}
\underset{x\in \mathcal X}{\text{maximize}} &\quad \sum^k_{i=1}\operatorname{Tr} (A^{\hh}_iB^{+}_iA_i) 
\end{align}
is equivalent to
\begin{align}
\label{QT prob:generalized inverse}
\underset{x \in \mathcal X, Y_1,\ldots,Y_k \in\mathbb C^{n\times m}}{\text{\ \ maximize}} &\quad \sum^k_{i=1}\operatorname{Tr}\!\big(2\Re\{Y_i^{\hh}A_i\}-Y_i^{\hh}B_iY_i\big) 
\end{align}
in the sense that the two problems have the same solution for $x$ and their optimal objective values are equal.
\end{theorem}

\subsection{Singular Denominators and Range Condition}

The range condition $(I_n-B_iB^+_i)A_i=0$ in Theorem \ref{thm:FP:generalized inverse} is inherited from the generalized
Schur complement \cite{Fuzhen_book}. We now clarify its meaning from the viewpoint of
singular matrix fractions.

For $B\succ 0$, the matrix fraction $A^\hh B^{-1}A$ is always well
defined. When $B\succeq 0$ is singular, however, the regularized quantity $A^\hh (B+\epsilon I)^{-1}A$ 
may diverge as $\epsilon$ tends to zero from above. The following lemma shows that
the range condition is exactly the condition under which the limiting matrix fraction remains finite.

\begin{lemma}
\label{lemma:limit matrix frac}
For $B\in\mathbb H_+^n$ and $A\in\mathbb C^{n\times m}$, the limit
\begin{equation}
    \lim_{\epsilon\downarrow 0}
    A^\hh(B+\epsilon I)^{-1}A
    \label{eq:regularized_fraction}
\end{equation}
exists as a finite Hermitian matrix if and only if
\begin{equation}
    \mathcal R(A)\subseteq \mathcal R(B), 
\end{equation}
which is equivalent to the condition
\begin{equation}
(I_n-BB^+)A=0.
    \label{eq:range_condition}
\end{equation}
When \eqref{eq:range_condition} holds, the limit is
\begin{equation}
    \lim_{\epsilon\downarrow 0}
    A^\hh(B+\epsilon I)^{-1}A
	=A^\hh B^+ A
	=A^\hh B^\dag A,
    \label{eq:regularized_limit_mp}
\end{equation}
where $B^\dag$ denotes the Moore-Penrose inverse of $B$. 
In other words, the generalized matrix fraction $A^\hh B^+A$ is independent of 
the choice of generalized inverse $B^+$ and is equal to $A^\hh B^\dag A$.

If $\mathcal R(A)\not\subseteq\mathcal R(B)$, then the limit diverges in
at least one direction; more specifically, there exists $u\in\mathbb C^m$ such
that
\begin{equation}
    \lim_{\epsilon\downarrow 0}
    u^\hh A^{\hh}(B+\epsilon I)^{-1}Au
    =
    +\infty.
    \label{eq:regularized_divergence}
\end{equation}
\end{lemma}

\begin{IEEEproof}
Let $r=\operatorname{rank}(B)$. Take the eigenvalue
decomposition of $B\succeq 0$ as
\begin{equation}
B=U_r\Lambda U_r^{\hh},
\end{equation}
where $\Lambda\in\mathbb H_{++}^r$ contains all the positive
eigenvalues of $B$, and the columns of $U_r$ form an orthonormal
basis of $\mathcal R(B)$. Let $U_0$ be an orthonormal basis of
$\mathcal N(B)$. Thus,
\begin{equation}
    \mathcal{R}(U_r) = \mathcal{R}(B),\qquad \mathcal{R}(U_0) = \mathcal{N}(B).
\end{equation}
Then $U=[U_r\ U_0]$ is unitary and
\begin{equation}
B
=
U
\begin{bmatrix}
\Lambda & 0\\
0 & 0
\end{bmatrix}
U^{\hh} .
\end{equation}
Write
\begin{equation}
U^{\hh}A=
\begin{bmatrix}
A_1\\
A_2
\end{bmatrix},
\end{equation}
or equivalently,
\begin{equation}
A=U_rA_1+U_0A_2 .
\end{equation}

We first show that $A_2=0$ is equivalent to
$\mathcal R(A)\subseteq \mathcal R(B)$.
Indeed, if $A_2=0$, then $A=U_rA_1$, so $\mathcal R(A)\subseteq\mathcal R(U_r) = \mathcal R(B)$. Conversely, if $\mathcal R(A)\subseteq\mathcal R(B)$,
then $A$ has no component along $\mathcal N(B)=\mathcal R(U_0)$, so $A_2=U_0^{\hh}A=0$.

Next, for any $\epsilon>0$,
\begin{equation}
\begin{aligned}
A^{\hh}(B+\epsilon I_n)^{-1}A
&=
A_1^{\hh}(\Lambda+\epsilon I_r)^{-1}A_1
+\frac{1}{\epsilon}A_2^{\hh}A_2 .
\end{aligned}
\end{equation}
If $A_2=0$, then the second term is zero, so
\begin{equation}
\lim_{\epsilon\downarrow0}
A^{\hh}(B+\epsilon I_n)^{-1}A
=
A_1^{\hh}\Lambda^{-1}A_1 .
\end{equation}
Moreover, since
\begin{equation}
B^\dag=U_r\Lambda^{-1}U_r^{\hh}
\end{equation}
and $A=U_rA_1$, we obtain
\begin{equation}
A_1^{\hh}\Lambda^{-1}A_1
=
A^{\hh}B^\dag A.
\end{equation}
Therefore, when $\mathcal R(A)\subseteq\mathcal R(B)$, the limit
exists as a finite Hermitian matrix and equals $A^{\hh}B^\dag A$. 

We next show that the same expression equals $A^HB^+A$ for any generalized
inverse $B^+$. The condition $\mathcal R(A)\subseteq\mathcal R(B)$ is equivalent to
the existence of a matrix $T$ such that
\begin{equation}
    A=BT.
    \label{eq:A_in_range_B}
\end{equation}
This is equivalent to
\begin{equation}
    (I_n-BB^+)A=0
    \label{eq:range_condition_ginv_equiv}
\end{equation}
for any generalized inverse $B^+$ of $B$, because by definition $BB^+B=B$. 
Under this condition, the
matrix fraction $A^{\hh}B^+A$ is independent of the particular choice of
generalized inverse. Indeed, if $A=BT$, then
\begin{equation}
    A^\hh B^+A
    =
    T^\hh BB^+BT
    =
    T^\hh BT,
    \label{eq:ginv_independence}
\end{equation}
which does not depend on the specific choice of $B^+$. Therefore,
under the range condition \eqref{eq:range_condition}, the generalized matrix fraction $A^\hh B^+A$
is well defined and coincides with the limiting expression
in \eqref{eq:regularized_limit_mp}, which is equal to $A^\hh B^+ A$.

It remains to consider the case $A_2\neq0$. In this case, there
exists a vector $u\in\mathbb C^m$ such that $A_2u\neq0$. Then
\begin{equation}
u^{\hh}A^{\hh}(B+\epsilon I)^{-1}Au
\geq
\frac{1}{\epsilon}\|A_2u\|_2^2
\to+\infty
\end{equation}
as $\epsilon$ tends to zero from above. Hence the limit cannot exist as a finite Hermitian matrix.
\end{IEEEproof}

\subsection{Generalized Quadratic Transform in Loewner Order}

Theorem \ref{thm:FP:generalized inverse} extends the quadratic transform to trace objectives involving
possibly singular denominator matrices. We now show that the same idea
can be further generalized beyond the trace objective. The key
observation is that, under the range condition, the quadratic transform
constructs a greatest element in the partial order defined by the positive semidefinite 
matrices, also known as the Loewner order.

For $B\succ 0$, the classical matrix quadratic transform is based on the identity
\begin{equation}
    A^\hh B^{-1}A
    =
    \max_V
    \left\{
        A^\hh V+V^\hh A-V^\hh BV
    \right\}.
    \label{eq:classical_matrix_qt}
\end{equation}
The next result extends this identity to the case $B\succeq 0$.

\begin{theorem}
\label{thm:FP:generalized inverse:Loewner}
Let $\mathcal X$ be a nonempty constraint set, and let
\begin{equation}
    A:\mathcal X\to\mathbb C^{n\times m},
    \qquad
    B:\mathcal X\to\mathbb H_+^n.
\end{equation}
Assume that for every $x\in\mathcal X$,
\begin{equation}
    \mathcal R(A(x))\subseteq \mathcal R(B(x)).
    \label{eq:theorem7_range_condition}
\end{equation}
Let $f:\mathbb H^m\to\mathbb R$ be nondecreasing with respect to the Loewner order, i.e.,
\begin{equation}
f(W_1) \ge f(W_2)\quad\text{whenever}\quad W_1\succeq W_2.
\label{eq:loewner_monotone_f}
\end{equation}
For each $x$, let $B(x)^+$ be any generalized inverse of $B(x)$.
Then, the matrix FP problem
\begin{equation}
    \operatorname*{maximize}_{x\in\mathcal X}\;
    f\left(A(x)^\hh B(x)^+A(x)\right)
    \label{eq:matrix_fp_ginv}
\end{equation}
is equivalent to
\begin{equation}
    \operatorname*{maximize}_{x\in\mathcal X,\;V\in\mathbb C^{n\times m}}\;
    f\left(\Gamma(x,V)\right)
    \label{eq:matrix_qt_ginv}
\end{equation}
where
\begin{equation}
\Gamma(x,V) = A(x)^\hh V + V^\hh A(x) - V^\hh B(x)V,  \label{eq:Q_def_ginv}
\end{equation}
in the sense that the two problems have the same optimal objective value and the same optimal solutions for $x$.  Moreover, for fixed $x$, the following is an optimal solution for $V$:
\begin{equation}
    V_0 = B(x)^+ A(x).
    \label{eq:theorem7_optimal_V}
\end{equation}
This solution satisfies $B(x) V_0 = A(x)$. In fact, if $f(\cdot)$ is strictly increasing in the Loewner order, then the optimal $V$ must all satisfy $B(x) V = A(x)$. 
\end{theorem}

\begin{IEEEproof}
Fix an arbitrary $x\in\mathcal X$, and omit the argument $x$ for $A(x)$ and $B(x)$ in the rest of the proof whenever doing so does not cause confusion. By the range condition
    \eqref{eq:theorem7_range_condition}, there exists a matrix $T$ such that
$A=BT$. Hence, for any generalized inverse $B^+$,
\begin{equation}
    BB^+A
    =
    BB^+BT
    =
    BT
    =
    A.
    \label{eq:BBplusA_equals_A}
\end{equation}
Now, let
\begin{equation}
    V_0=B^+A
\end{equation}
for which we have
\begin{equation}
    BV_0=A.
\end{equation}
Since $BV_0=A=BT$,
\begin{equation}
    V_0^\hh BV_0
    =
    T^\hh BT
    =
    A^\hh B^+A.
    \label{eq:V0_B_V0}
\end{equation}
This allows us to complete the square for $\Gamma(x,V)$ as follows:  
\begin{align}
    \Gamma(x,V)
    &=
    A^{\hh}V+V^{\hh}A-V^{\hh}BV  \notag\\
    &=
    A^{\hh}B^+A
    -
    (V-V_0)^{\hh}B(V-V_0),
    \label{eq:loewner_completion_square}
\end{align}
where we have utilized \eqref{eq:ginv_independence} and \eqref{eq:V0_B_V0}. 
Since $B\succeq 0$, we have
\begin{equation}
    (V-V_0)^{\hh}B(V-V_0)\succeq 0.
    \label{eq:optimal_V}
\end{equation}
Therefore,
\begin{equation}
    \Gamma(x,V)
    \preceq
    A^{\hh}B^+A,
    \qquad
    \forall V\in\mathbb C^{n\times m},
    \label{eq:Q_loewner_upper_bound}
\end{equation}
with equality achieved at $V=V_0=B^+A$. Thus, $V_0$ is one particular choice of optimal $V$. 

Now, we have established that $A^\hh B^+A$ is the greatest element among
the matrices $\Gamma(x,V)$ in terms of the Loewner order. Since $f$ is nondecreasing with respect to the Loewner order, we have
\begin{equation}
    f(\Gamma(x,V))
    \leq
    f(A^\hh B^+A),
    \qquad
    \forall V.
\end{equation}
The equality is attained by $V=B^+A$. Hence for any fixed $x$,
\begin{equation}
    f(A^\hh B^+A)
    =
    \max_{V\in\mathbb C^{n\times m}}
    f\left(A^\hh V+V^\hh A-V^\hh BV\right).
    \label{eq:fixed_x_matrix_qt_ginv}
\end{equation}
Maximizing both sides over $x\in\mathcal X$ proves the equivalence
between \eqref{eq:matrix_fp_ginv} and \eqref{eq:matrix_qt_ginv}.

Finally, if $f$ is increasing, then any optimal $V$ must be such that $BV = B V_0$ in order to achieve equality in \eqref{eq:optimal_V}, so any optimal $V$ must satisfy $BV = A$.
\end{IEEEproof}

\begin{remark}
\label{remark:sol of V}
The optimal auxiliary variable $V$ in Theorem \ref{thm:FP:generalized inverse:Loewner}
may not be unique. In particular, when $B(x)$ is singular, every matrix of the form
\begin{equation}
    V=B(x)^+A(x)+U
    \label{eq:V_nonunique_ginv}
\end{equation}
along with some matrix $U$ satisfying
\begin{equation}
    B(x)U=0
\end{equation}
attains the Loewner upper bound in \eqref{eq:Q_loewner_upper_bound}.
\end{remark}


\begin{remark}
If $B(x)\succ 0$, then the range condition is automatic, the
generalized inverse reduces to the ordinary inverse, and the optimal
auxiliary variable is unique:
\begin{equation}
    V=B(x)^{-1}A(x).
\end{equation}
In this case, Theorem \ref{thm:FP:generalized inverse:Loewner} reduces to the classical matrix quadratic transform \cite{FP_SPM}.
\end{remark}

\section{Minimax Duality for Gaussian Vector Broadcast Channel via Generalized FP}
\label{sec:duality}

To illustrate the usefulness of generalized FP theory, this section applies the generalized FP to the least-favorable-noise minimax characterization of the Gaussian vector BC sum capacity \cite{Yu2006}. The generalized quadratic transform in Theorem \ref{thm:FP:generalized inverse:Loewner} is able to easily handle the case of possibly singular noise covariance and can be used to directly establish the BC-MAC sum-capacity duality.



\subsection{Channel Models}

\subsubsection{Broadcast Channel}
Consider a $k$-user Gaussian vector BC, where the common transmitter has $M$ transmit antennas and user $i$ has $N_i$ receive antennas. The BC channel model is
\begin{equation}
    Y_i = G_iX+Z_i,\qquad i=1,\ldots,k,
    \label{eq:bc_channel}
\end{equation}
where $G_i\in\mathbb C^{N_i\times M}$, $X\in\mathbb C^M$, and $Z_i\sim\mathcal{CN}(0,I_{N_i})$. Define the stacked channel matrix
\begin{equation}
    H=
    \begin{bmatrix}
        G_1\\
        \vdots\\
        G_k
    \end{bmatrix}
    \in\mathbb C^{N\times M},
    \label{eq:stacked_channel}
\end{equation}
where $N=\sum_{i=1}^k N_i$. Assume that $H\neq 0$. The input covariance matrix $S=\mathbb E[XX^\hh]$ is subject to the normalized sum-power constraint $\operatorname{Tr}(S)\le1$.

\subsubsection{Multiple-Access Channel}

The sum capacity of the Gaussian vector BC can be characterized in terms of the sum capacity of a reciprocal Gaussian vector MAC \cite{Pramod_BC, Sriram_BC}. This reciprocal MAC is obtained by reversing the channel matrices. Thus, the user $i$ in the MAC has $N_i$ transmit antennas, and the receiver has $M$ antennas. The reciprocal MAC model is
\begin{equation}
    \widetilde Y=\sum_{i=1}^k G_i^\hh \widetilde X_i+
    \widetilde Z=H^\hh \begin{bmatrix}
        \widetilde X_1\\
        \vdots\\
        \widetilde X_k
    \end{bmatrix}+ \widetilde Z,
    \label{eq:mac_channel}
\end{equation}
where $\widetilde X_i\in\mathbb C^{N_i}$
and $\widetilde Z\sim\mathcal{CN}(0,I_M)$.
The transmit covariance matrix of the $i$-th MAC user is $W_i=\mathbb E[\widetilde X_i\widetilde X_i^{\hh}]$. The duality between BC and MAC is established under a sum power constraint $\operatorname{Tr}(W) \le 1$,
where $W=\operatorname{blkdiag}(W_1,\ldots,W_k)$.

\subsection{Minimax Characterization of BC Sum Capacity}
\label{subsec:review minimax}

A key result in the analysis of the Gaussian vector BC is that its sum capacity admits the following characterization:
\begin{equation}
\begin{aligned}
    \min_{D}\ \max_{S}\quad
    &\lim_{\epsilon\downarrow 0}
    \log\left|
        I_N+HSH^{\hh}(D+\epsilon I_N)^{-1}
    \right|\\
    \text{\ \ subject to}\quad
    &\,S\succeq 0,\quad D\succeq 0\\
    &\operatorname{Tr}(S)\leq 1,\\
    &\,\mathcal P_{\rm blk}(D)=I_{N},
\end{aligned}
\label{eq:lfn_minimax}
\end{equation}
where notationally, for an $N\times N$ matrix partitioned according to dimensions
$N_1,\ldots,N_k$, we use a projection operator 
\begin{equation}
\mathcal P_{\rm blk}(D)=
\operatorname{blkdiag}(D_{11},\ldots,D_{kk}),
\end{equation}
to denote its block diagonal terms $D_{ii}=\mathbb E[Z_iZ^\hh_i]$.
Here, $D$ is the covariance matrix of an
artificially correlated Gaussian noise vector across the receivers called
\emph{least-favorable noise} \cite{Yu_BC}. The constraint
$\mathcal P_{\rm blk}(D)=I_N$ states that the artificial noise vectors may be
arbitrarily correlated across users, while the covariance matrices of their marginals
remain equal to the identity matrix.
The limiting expression in the objective takes into account the possibility that $D$ may be singular.  

The fact that \eqref{eq:lfn_minimax} is an outer bound of 
the sum capacity of the BC is due to \cite{Sato_BC}. 
The achievability is established in \cite{Yu_BC} based on 
a decision-feedback equalization technique and in \cite{Sriram_BC, Pramod_BC}
based on the uplink-downlink duality. In \cite{Yu2006}, an alternative proof
of sum-capacity duality based directly on the minimax characterization
\eqref{eq:lfn_minimax} is provided, but a complication arises in its derivation in that 
the least-favorable noise covariance $D$ may be singular. The minimax formulation provides a natural application scenario for generalized FP, 
because it is capable of handling generalized inverse of the possibly singular 
noise covariance. 
In this paper, we re-derive the duality result based on generalized FP. 

Before utilizing the tool of generalized FP, we first briefly review the minimax duality 
result in \cite{Yu2006}
under a set of simplifying assumptions that 
\begin{enumerate}[(i)]
    \item the stacked channel matrix $H$ is square and nonsingular;
    \item the optimal $D$ is invertible;
    \item the optimal $S$ is invertible.
\end{enumerate}
Then, the Karush-Kuhn-Tucker (KKT) conditions of the minimax formulation \eqref{eq:lfn_minimax} reduce to
\begin{equation}
    H^\hh (HSH^\hh+D)^{-1}H= \lambda I_M,
    \label{eq:direct_kkt_S}
\end{equation}
and
\begin{equation}
    D^{-1}-(HSH^\hh+D)^{-1}=\Delta,
    \label{eq:direct_kkt_D}
\end{equation}
where $\lambda \ge 0$ is the multiplier associated with the transmit covariance constraints, 
and $\Delta=\mathcal P_{\rm blk}(\Delta)$ is the block-diagonal
multiplier associated with the constraint $\mathcal P_{\rm blk}(D)=I_N$.

It is shown in \cite{Yu2006} that \eqref{eq:direct_kkt_S} and \eqref{eq:direct_kkt_D} can be solved
explicitly. From \eqref{eq:direct_kkt_D}, we have
\begin{equation}
    (HSH^\hh+D)^{-1}=D^{-1}-\Delta.
\end{equation}
Substituting this into \eqref{eq:direct_kkt_S} gives
\begin{equation}
    H^\hh D^{-1}H
    =
    H^\hh\Delta H+  \lambda I_M. 
    \label{eq:direct_HDinvH}
\end{equation}
Thus,
\begin{equation}
    H\left(H^\hh\Delta H+ \textstyle  \lambda I_M \right)^{-1}H^\hh
    =
    D.
    \label{eq:direct_D_solution}
\end{equation}
Moreover, \eqref{eq:direct_kkt_S} implies
\begin{equation}
    HSH^\hh+D
    =
    \frac{1}{\lambda}H H^\hh.
\end{equation}
Combining this identity with \eqref{eq:direct_D_solution}, we obtain
\begin{equation}
  \frac{1}{\lambda}I_M
    -
    \left(H^\hh\Delta H+\lambda I_M \right)^{-1}
    =
    S.
    \label{eq:direct_S_solution}
\end{equation}
Substituting \eqref{eq:direct_HDinvH} and
\eqref{eq:direct_S_solution} into the BC objective yields
\begin{equation}
\log\left|I_M+SH^\hh D^{-1}H\right|
    = 
\log\left|I_M + \frac{1}{\lambda}H^\hh\Delta H \right|.
    \label{eq:direct_kkt_logdet_identity}
\end{equation}
The right-hand side above can be seen as the dual MAC capacity expression with 
\begin{equation} W = \frac{\Delta}{\lambda} \label{eq:dual_W_delta} \end{equation}
as the transmit covariance matrix with a normalized noise covariance, thus establishing uplink-downlink duality.

Furthermore, the diagonal part of \eqref{eq:direct_D_solution} is the KKT condition
of the following MAC sum-capacity maximization problem:
\begin{equation}
\begin{aligned}
\underset{W}{\text{maximize}}\quad
&
\log\left|
I_M+H^{\hh} W H\right|\\
    \text{subject to}\quad
    &\, W=\operatorname{blkdiag}(W_1,\ldots,W_k) \succeq 0 \\
    &\operatorname{Tr}(W)\le 1.
\end{aligned}
\label{eq:MAC_minimax}
\end{equation}
The above is the central algebraic derivation behind the Lagrangian duality
interpretation of uplink-downlink duality in \cite{Yu2006}. It shows that
the sum-capacity maximizing transmit covariance matrices of the dual MAC are 
precisely the optimal dual variables associated with the least favorable noises
in the BC sum-capacity maximization problem. 

However, the derivation
above makes the simplifying assumptions (i)-(iii) above. In \cite{Yu2006},
considerable effort is devoted to the analysis of the signal and noise
subspaces to derive the result for the case that $D$ is singular.  In the following, 
we show that the generalized FP can be utilized to directly handle
the possibly singular $D$.

\subsection{Generalized FP Reformulation}

Consider the BC sum-capacity objective \eqref{eq:lfn_minimax}.
If the range condition $\mathcal R(H)\subseteq \mathcal R(D)$ holds, then 
by Lemma \ref{lemma:limit matrix frac}, the regularized objective of \eqref{eq:lfn_minimax} is equal to
\begin{align}
	\log\left|I_M+\sqrt{S}H^{\hh}D^\dag H \sqrt{S}\right|.
    \label{eq:lfn_limit_pinv}
\end{align}
In this case, by defining
\begin{equation}
    \Gamma
    =
    H^\hh V+V^\hh H-V^\hh DV,
    \label{eq:Gamma_def_lfn}
\end{equation}
problem \eqref{eq:lfn_minimax} can be recast 
using Theorem \ref{thm:FP:generalized inverse:Loewner} as
\begin{equation}
\begin{aligned}
    \min_D\ \max_{S,V}\quad
	&\log\left|I_M+ \sqrt{S}\Gamma \sqrt{S}\right|\\
    \text{subject to}\quad
    &\,S\succeq 0,\quad D\succeq 0\\
    &\operatorname{Tr}(S)\leq 1\\
    &\,\mathcal P_{\rm blk}(D)=I_{N}.
\end{aligned}
\label{eq:lfn_qt_reformulation}
\end{equation}
On the other hand, if $\mathcal{R}(H)\not\subseteq\mathcal{R}(D)$, the inner maximization in \eqref{eq:lfn_minimax} is unbounded, so such a $D$ cannot be optimal in the outer minimization.
Thus, it is without loss of generality to assume $\mathcal R(H)\subseteq \mathcal R(D)$ and to specialize the generalized inverse as the Moore-Penrose pseudo-inverse in this problem. 

The above generalized-FP reformulation avoids the difficulty of having to differentiate with respect to a possibly singular $D$ in the pseudo-inverse in \eqref{eq:lfn_limit_pinv}---which is nonsmooth in $D$ when it changes rank.
Instead, \eqref{eq:lfn_qt_reformulation} provides a smooth representation in the lifted space $(D,V)$. This is different from the technique used in \cite{Yu2006} that makes the objective function smooth in $D$ by reducing the channel to an active subspace wherein the reduced noise covariance is nonsingular.

We now write the KKT conditions for the BC sum-capacity problem in this new form. The stationarity conditions of \eqref{eq:lfn_qt_reformulation} with respect to $S$ and $D$ are
\begin{equation} 
\label{S-stationarity}
\Gamma(I_M+ S \Gamma )^{-1} + \Phi = \lambda I_M,
\end{equation}
and
\begin{equation}
\label{D-stationarity}
V(I_M+ S \Gamma)^{-1}SV^\hh + \Omega = \Delta,
\end{equation}
where $\Phi\succeq0$ is the Lagrange multiplier for $S\succeq0$; $\lambda\ge0$ is the multiplier for $\operatorname{Tr}(S)\le1$; $\Omega\succeq0$ is the Lagrange multiplier for $D\succeq0$; and $\Delta$ is a block-diagonal matrix, whose block diagonal terms are the Lagrange multipliers 
corresponding to $D_{ii}=I_{N_i}$. 
Note that the first terms in both \eqref{S-stationarity} and \eqref{D-stationarity} are Hermitian matrices.
This can be seen by moving a factor $\sqrt{\Gamma}$ and $\sqrt{S}$ in and out of the inverse term from the left or right. At the optimum, we have $\Gamma \succeq 0$, as will be verified shortly in \eqref{eq:Gamma_identities}. 

Furthermore, the complementary slackness conditions are
\begin{equation}
\label{complementary slackness:BC}
\operatorname{Tr}(\Phi S)=0,\quad \operatorname{Tr}(\Omega D) = 0, \quad \lambda\left(\operatorname{Tr}(S)-1\right)=0.
\end{equation}
Note that \eqref{complementary slackness:BC} implies $\Phi S=S\Phi=0$ and $\Omega D=D\Omega=0$.

The new ingredient here is the auxiliary variable $V$.  By Theorem \ref{thm:FP:generalized inverse:Loewner}, one optimal $V$ is $V = D^+ H$, which satisfies 
\begin{equation}
    DV=H.
    \label{eq:DV_equals_H}
\end{equation}
This is in fact an alternative expression of the range condition $\mathcal{R}(H)\subseteq\mathcal{R}(D)$. 
For such optimal $V$, we have
\begin{equation}
    \Gamma = H^{\hh}D^{\dag} H = H^{\hh}V = V^{\hh}H = V^{\hh}DV.
    \label{eq:Gamma_identities}
\end{equation}
This key identity is useful in the subsequent developments.

For the reciprocal MAC problem in \eqref{eq:MAC_minimax},
the stationarity condition with respect to $W$ is
\begin{equation}
\label{W-stationarity}
	\mathcal{P}_{\rm blk}\Bigl(H( I_M+H^\hh W H)^{-1}H^\hh\Bigr)+\Psi = \mu I_N,
\end{equation}
where $\Psi \succeq0$ and its block diagonal terms are the Lagrange multipliers for $W_i \succeq0$; $\mu\ge0$ is the Lagrange multiplier for $\operatorname{Tr}(W)\le 1$. 

The corresponding complementary slackness conditions are
\begin{equation}
\label{complementary slackness:MAC}
    \operatorname{Tr}(\Psi W)=0,\quad \mu\bigl(\operatorname{Tr}(W)-1\bigr)=0.
\end{equation}
Since $\Psi\succeq0$ and $W\succeq0$, \eqref{complementary slackness:MAC} implies $\Psi W=W\Psi=0$. 

We now show that the BC-KKT conditions 
\eqref{S-stationarity}-\eqref{complementary slackness:BC}
and the MAC-KKT conditions 
\eqref{W-stationarity}-\eqref{complementary slackness:MAC}
are equivalent without having to make the simplifying assumptions (i)-(iii). 
Because $S$ and $D$ are no longer assumed to be positive definite, 
we need to explicitly characterize the dual variables $\Phi$ and $\Omega$.

\subsection{Equivalence of BC-KKT and MAC-KKT Conditions}
\label{subsec:BC to MAC}

Let $(S,V,D,\Phi,\lambda,\Omega,\Delta)$ be a primal-dual KKT-tuple of the BC problem \eqref{eq:lfn_qt_reformulation}, where $DV=H$. We show that the following construction gives the KKT-tuple of the MAC problem \eqref{eq:MAC_minimax}:
\begin{align}
W &= \frac{\Delta}{\lambda}, \label{construct W} \\
\Psi &= \lambda I_N-\mathcal{P}_{\rm blk}\bigl(H (I_M+H^\hh W H)^{-1}H^\hh\bigr),\label{construct Psi}\\
\mu &= \lambda.\label{construct mu}
\end{align}

Here, \eqref{construct W} is inspired by \eqref{eq:dual_W_delta}. 
Since $\Delta$ is block diagonal, so is $W$. Furthermore, $\Omega \succeq 0$ implies $\Delta \succeq 0$. Together with $\lambda > 0$, we have $W \succeq 0$. Moreover, 
\eqref{construct Psi} is constructed to satisfy the stationarity condition \eqref{W-stationarity} for the MAC. 

To show that this set of $(W, \Psi, \mu)$ satisfies the MAC-KKT condition, it remains to check that (a) the power constraint is satisfied; (b) $\Psi \succeq 0$; and (c) the complementary slackness condition $\operatorname{Tr}(\Psi W)=0$ is met. 

First, to evaluate the power constraint $\operatorname{Tr}(W)$, we observe that the following chain of equalities holds:
\begin{align}
&\operatorname{Tr}(\Delta) = \operatorname{Tr}(\Delta D) = \operatorname{Tr}\left(V(I_M+ S \Gamma)^{-1}SV^\hh D\right)\notag\\
&= \operatorname{Tr}\left((I_M+ S \Gamma)^{-1}SV^\hh DV\right) = \operatorname{Tr}((I_M+S\Gamma)^{-1}S\Gamma)\notag\\
& = \operatorname{Tr}(\Gamma(I_M+S\Gamma)^{-1}S)=\lambda\operatorname{Tr}(S),
\label{lambda Tr(s)}
\end{align}
where the first line uses the fact that $\Delta$ is block diagonal and $D$ has identity matrices on its block diagonal, and further we make use of the $D$-stationarity condition \eqref{D-stationarity} and $\Omega D=0$; the second line uses \eqref{eq:Gamma_identities}; and the last line uses the $S$-stationarity condition \eqref{S-stationarity} and $\Phi S=0$. Thus,
\begin{equation}
\label{tr W = tr S}
    \operatorname{Tr}(W) = \operatorname{Tr}(S)\le 1.
\end{equation}
Note that the MAC power complementary slackness condition follows from its BC counterpart.

Next, to check that $\Psi \succeq 0$, 
we rewrite  $\Psi$ in \eqref{construct Psi} using 
$\mathcal{P}_{\rm blk}(D) = I_N$ and $\Delta=\lambda W$ as 
\begin{equation}
\label{Psi new}
\Psi = \lambda\mathcal{P}_{\rm blk}\Bigl(D-H(\lambda I_M+H^\hh \Delta H)^{-1}H^\hh\Bigr).
\end{equation}
Since $\lambda>0$, to show $\Psi \succeq 0$, it is sufficient to show
\begin{equation}
\label{Psi remove blk}
 D-H(\lambda I_M+H^\hh \Delta H)^{-1}H^\hh\succeq0,
\end{equation}
which, by the Schur-complement LMI, is equivalent to
\begin{equation}
\label{Psi Schur complement}
\begin{bmatrix}
    D & H\\
    H^\hh & \lambda I_M+H^\hh\Delta H
\end{bmatrix}
\succeq0.
\end{equation}
We evaluate the bottom-right entry of the above matrix as
\begin{align}
\lambda I_M + H^\hh \Delta H &= \lambda I_M + \Gamma (I_M+S\Gamma)^{-1}S \Gamma\notag\\
&=\lambda I_M + (\lambda I_M-\Phi)S\Gamma\notag\\
&=\lambda I_M + \lambda S\Gamma\notag\\
&=\Gamma + \Phi\notag\\
&= V^\hh D V + \Phi,\label{lambda S Gamma}
\end{align}
where the first line uses the $D$-stationarity condition \eqref{D-stationarity} together with \eqref{eq:Gamma_identities}
and $\Omega H=\Omega DV=0$; the second line uses the $S$-stationarity condition \eqref{S-stationarity}; the third line uses $\Phi S=0$; the fourth line uses the $S$-stationarity condition \eqref{S-stationarity} (with $(I_M+S\Gamma)$ 
multiplied on both sides) together with $\Phi S=0$; and the last line uses \eqref{eq:Gamma_identities}.

Using \eqref{lambda S Gamma} and $H=DV$ in \eqref{eq:DV_equals_H}, we can verify the positive semidefiniteness in \eqref{Psi Schur complement} as
\begin{align}
\label{Psi>=0:2}
&\begin{bmatrix}
    D & H\\
    H^\hh & \lambda I_M+H^\hh\Delta H
\end{bmatrix}
=
\begin{bmatrix}
    D & DV\\
    V^\hh D & V^\hh D V + \Phi
\end{bmatrix}\notag\\
&=
\begin{bmatrix}
    I_N\\ V^\hh
\end{bmatrix}
D
\begin{bmatrix}
    I_N & V
\end{bmatrix}
+
\begin{bmatrix}
    0 & 0\\
    0 & \Phi
\end{bmatrix}
\succeq0,
\end{align}
thereby establishing $\Psi\succeq0$. 

Finally, to check that  $\operatorname{Tr}(\Psi W)=0$, we have
\begin{align}
    \operatorname{Tr}(\Psi W) &= \lambda\operatorname{Tr}(W)-\lambda\operatorname{Tr}(H(\lambda I_M+H^\hh \Delta H)^{-1}H^\hh W)\notag\\
    &= \lambda\operatorname{Tr}(W)-\lambda\operatorname{Tr}((\lambda I_M+H^\hh \Delta H)^{-1}H^\hh W H)\notag\\
    &= \lambda\operatorname{Tr}(W)-\lambda\operatorname{Tr}((\lambda I_M+H^\hh \Delta H)^{-1}S\Gamma)\notag\\
    &=  \lambda\operatorname{Tr}(W)-\lambda\operatorname{Tr}(S) = 0,
\end{align}
where the third line uses $H^\hh \Delta H=\lambda S\Gamma$ from \eqref{lambda S Gamma} together with $\Delta=\lambda W$, and the last line uses $\Gamma = \lambda I_M+H^\hh\Delta H-\Phi$ also from \eqref{lambda S Gamma} together with $S\Phi=0$. 

Thus, the constructed $(W,\Psi,\mu)$ satisfy the KKT conditions of the reciprocal MAC problem.
The key observation here is that the above derivation does not assume that $S$, $D$ and $H$ are invertible. The use of the auxiliary variable $V$ allows the rank-deficient cases to be handled automatically.

\subsection{Sum Capacities of BC and MAC}

As a final step, we can also check that the sum capacities of the BC and MAC are equal. 
Let $(S^\star,D^\star)$ be an optimal saddle point of problem \eqref{eq:lfn_minimax}, and choose $V^\star$ such that $D^\star V^\star=H$. By Theorem \ref{thm:FP:generalized inverse:Loewner}, the BC sum capacity is given by
\begin{align}
\mathcal C_{\rm BC}=\log\left|I_M+S^\star\Gamma^\star\right|,
\end{align}
where 
\begin{equation}
 \Gamma^\star = H^\hh (D^\star)^\dag H = V^{\hh}D^\star V. 
\end{equation}
Let $W^\star$ be the corresponding MAC covariance obtained from the KKT mapping \eqref{construct W}-\eqref{construct mu}. It follows from the third line of \eqref{lambda S Gamma} that the MAC sum capacity satisfies
\begin{align}
\mathcal C_{\rm MAC} & = \log\left|I_M+H^\hh W^\star H\right| \nonumber \\
& = \log\left|I_M+S^\star \Gamma^\star\right| \nonumber \\
& = \mathcal C_{\rm BC}
\end{align}
Hence, the generalized FP formulation establishes the BC–MAC sum-capacity duality regardless of whether the least-favorable noise covariance $D^\star$ is singular.

\section{Conclusion}
\label{sec:conclude}

This paper shows that the quadratic transform technique for solving FP is closely connected to the Schur complement in the sense that they can be derived from each other. The auxiliary variable in FP is the dual variable associated with an LMI in the Schur complement reformulation of the FP; it can also be interpreted in the context of an MMSE estimation problem. This connection leads to new generalizations and an alternative proof of the Schur-complement determinant formula. 

Moreover, we apply the generalized inverse viewpoint to the least-favorable-noise minimax formulation of Gaussian vector broadcast channels. The possible singularity of the least-favorable noise covariance makes direct inverse-based KKT analysis inconvenient. We show that the generalized quadratic transform provides an auxiliary-variable representation of the singular matrix fraction and directly yields the reciprocal MAC covariance from the KKT multipliers. This formulation avoids explicit inversion of the noise covariance and recovers the uplink-downlink duality relation without separate support-space reductions.

\bibliographystyle{IEEEtran}
\bibliography{IEEEabrv,refs}

@STRING{IEEE_J_SP         = "{IEEE} Trans. Signal Process."}

@STRING{IEEE_J_COM        = "{IEEE} Trans. Commun."}

@STRING{IEEE_J_WCOM       = "{IEEE} Trans. Wireless Commun."}

@STRING{IEEE_J_IT         = "{IEEE} Trans. Inf. Theory"}

@STRING{IEEE_J_IOT        = "{IEEE} Internet Things J."}

@STRING{IEEE_M_SP         = "{IEEE} Signal Process. Mag."}

@STRING{IEEE_ISIT        = "Proc. {IEEE} Int. Symp. Inf. Theory (ISIT)"}

@STRING{IEEE_ICASSP        = "Proc. {IEEE} Int. Conf. Acoust., Speech, Signal Process. (ICASSP)"}

@article{Yu2006,
  author  = {Wei Yu},
  title   = {Uplink--Downlink Duality via Minimax Duality},
  journal = {IEEE Trans. Inf. Theory},
  volume  = {52},
  number  = {2},
  pages   = {361--374},
  month   = feb,
  year    = {2006}
}

@ARTICLE{attiah2026,
  author={Attiah, Kareem M. and Yu, Wei},
  journal= {IEEE J. Sel. Areas Inf. Theory}, 
  title={Uplink--Downlink Duality for Beamforming in Integrated Sensing and Communications}, 
  year={2026},
  volume={7},
  number={},
  pages={354-372},
  month={May}}

@ARTICLE{Yu_BC,
  author={Wei Yu and Cioffi, J.M.},
  journal = {IEEE Trans. Inf. Theory},
  title={Sum capacity of {Gaussian} vector broadcast channels}, 
  year={2004},
  volume={50},
  number={9},
  pages={1875-1892}
}

@inproceedings{ShenAttiahChenYu2006,
  author    = {Kaiming Shen and Khaled M. Attiah and Yihong Chen and Wei Yu},
  title     = {Connections Between Quadratic Transform for Fractional Programming and {Schur} Complement},
  booktitle = IEEE_ISIT,
  month     = jun,
  year      = {2026}
}

@ARTICLE{Sriram_BC,
  author={Vishwanath, S. and Jindal, N. and Goldsmith, A.},
  journal = {IEEE Trans. Inf. Theory},
  title={Duality, achievable rates, and sum-rate capacity of {Gaussian MIMO} broadcast channels}, 
  year={2003},
  volume={49},
  number={10},
  pages={2658-2668}
}

@ARTICLE{Pramod_BC,
  author={Viswanath, P. and Tse, D.N.C.},
  journal = {IEEE Trans. Inf. Theory},
  title={Sum capacity of the vector {Gaussian} broadcast channel and uplink–downlink duality}, 
  year={2003},
  volume={49},
  number={8},
  pages={1912-1921}
}

@ARTICLE{Sato_BC,
  author={Sato, H.},
  journal = {IEEE Trans. Inf. Theory},
  title={An outer bound to the capacity region of broadcast channels}, 
  year={1978},
  volume={24},
  number={3},
  pages={374-377}
}

@book{boyd2004convex,
  title={Convex optimization},
  author={Boyd, Stephen and Vandenberghe, Lieven},
  year={2004},
  publisher={Cambridge University Press}
}

@INPROCEEDINGS{10097000,
  author={Zhu, Minghe and Li, Lei and Xia, Shuqiang and Chang, Tsung-Hui},
  booktitle=IEEE_ICASSP,
  title={Information and Sensing Beamforming Optimization for Multi-User Multi-Target {MIMO ISAC} Systems}, 
  year={2023},
  volume={},
  number={},
  pages={},
  month = jun,
  doi={10.1109/ICASSP49357.2023.10097000}}

@article{max_min_FP,
  author={Chen, Yannan and Zhao, Licheng and Shen, Kaiming},
  journal=IEEE_J_SP, 
  title={Mixed Max-and-Min Fractional Programming for Wireless Networks}, 
  year={2023},
  volume={72},
  number={},
  pages={337-351},
  month=dec
}

@article{weizhang_TWC,
  author={Y. Tian and D. Wang and C. Huang and W. Zhang},
  journal=IEEE_J_WCOM, 
  title={Performance Trade-Off of Integrated Sensing and Communications for Multi-User Backscatter Systems}, 
  year={2024},
  volume={23},
  number={11},
  pages={17310-17323},
  month=nov
}

@article{Yuan_IOTJ,
  author={P. Liu and S. Xu and S. Tang and X. Wang and F. Xia and W. Yuan},
  journal=IEEE_J_IOT, 
  title={Sensing-Assisted Secure Communications: A Rate-Splitting Approach}, 
  year={2025},
  volume={12},
  number={20},
  pages={42876-42890},
  month=oct
}

@article{Wang_TWC,
  author={M. Yuan and D. He and H. Yin and H. Wang and F. Liu and Z. Wang},
  journal=IEEE_J_WCOM, 
  title={Hybrid Beamforming for mm{W}ave Integrated Sensing and Communication With Multi-Static Cooperative Localization}, 
  year={2025},
  volume={25},
  number={},
  pages={771-786},
  month=jul
}

@article{Ng_TWC,
  author={J. Zou and S. Sun and C. Masouros and Y. Cui and Y.-F. Liu and D. Ng},
  journal=IEEE_J_COM, 
  title={Energy-Efficient Beamforming Design for Integrated Sensing and Communications Systems}, 
  year={2024},
  volume={72},
  number={6},
  pages={3766-3782},
  month=jun
}

@article{FP_SPM,
  author={K. Shen and W. Yu},
  journal=IEEE_M_SP, 
  title={Quadratic Transform for Fractional Programming in Signal Processing and Machine Learning}, 
  year={2025},
  volume={42},
  number={3},
  pages={14-34},
  month=may
}

@article{dinkelbach1967nonlinear,
  title={On nonlinear fractional programming},
  author={Dinkelbach, W.},
  journal={Manage. Sci.},
  volume={13},
  number={7},
  pages={492--498},
  month=Mar,
  year={1967}
}

@article{shen2018fractional1,
  author = {Shen, K. and Yu, W.},
  journal = IEEE_J_SP,
  number = {10},
  pages = {2616--2630},
  title = {Fractional programming for communication systems---{Part} {\uppercase\expandafter{\romannumeral1}}: Power control and beamforming},
  volume = {66},
  month = Mar,
  year = {2018}
  }

@inproceedings{Attiah_ISIT24,
author = {K. M. Attiah and Wei Yu},
booktitle =IEEE_ISIT,
pages = {},
month =jul,
title = {Beamforming Design for Integrated Sensing and Communications Using Uplink--Downlink Duality},
year = {2024}
}

@article{Sion_minimax,
author = {M. Sion},
journal ={Pacific J. Math.},
pages = {171-176},
title = {On general minimax theorems},
year = {1958}
}

@article{Schur,
author = {J. Schur},
journal ={J. Reine Angew. Math.},
pages = {205-232},
title = {\"{U}ber {P}otenzreihen, die im {I}nneren des {E}inheitskreises beschr\"{a}nkt sind},
year = {1917}
}

@book{Fuzhen_book,
  title={The {S}chur Complement and Its Applications},
  author={F. Zhang},
  publisher={Springer},
  year={2005}
}

@book{Wilde_book,
  title={Quantum Information Theory},
  author={M. M. Wilde},
  edition={2nd},
  publisher={Cambridge University Press},
  year={2017}
}

@article{WZ_TIT,
  title={The rate--distortion function for source coding with side information at the decoder},
  author={A. Wyner and J. Ziv},
  journal=IEEE_J_IT,
  volume={22},
  number={1},
  pages={1--10},
  year={1976},
  month=jan
}

@article{Guo_TIT,
  title={Mutual information and minimum mean-square error in {G}aussian channels},
  author={D. Guo and S. Shamai and S. Verd\'{u}},
  journal=IEEE_J_IT,
  volume={51},
  number={4},
  pages={1261--1282},
  year={2005},
  month=apr
}

@inproceedings{WY_ISIT,
  title={Writing on colored paper},
  author={W. Yu and Sutivong, A. and Julian, D. and Cover, T.M. and M. Chiang},
  booktitle=IEEE_ISIT,
  pages={322},
  year={2001},
  month=jun
}

\end{document}